\documentclass[graybox]{svmult}

\usepackage{mathptmx}       
\usepackage{helvet}         
\usepackage{courier}        
\usepackage{type1cm}        
\usepackage{makeidx}         
\usepackage{graphicx}        
\usepackage{multicol}        
\usepackage[bottom]{footmisc}

\usepackage{amsmath}
\usepackage{amssymb}
\usepackage{amsfonts}
\usepackage{dsfont}
\usepackage{bm}

\makeindex             

\newcommand{\z}{{\mathbb Z}}

\newcommand{\R}{\mathbb{R}}

\newcommand{\C}{\mathbb{C}}

\newcommand{\ML}{\mathcal{M}}

\newcommand{\Ex}{\mathbb{E}}
\DeclareMathOperator*{\corr}{corr}

\begin{document}

\title*{Multi-Asset Option Pricing with Exponential L\'evy Processes and the Mellin Transform}
\author{D.J. Manuge}
\institute{D.J. Manuge \at University of Guelph, 50 Stone Rd. E., Guelph, Canada. \email{dmanuge@uoguelph.ca}}
%
%
\maketitle

\abstract*{Exponential L\'evy processes have been used for modelling financial derivatives because of their ability to exhibit many empirical features of markets. Using their multidimensional analog, a general analytic pricing formula is obtained, allowing for the direct valuation of multi-asset options on $n \in \z^+$ risky assets. By providing alternate expressions for multi-asset option payoffs, the general pricing formula can reduce to many popular cases, including American basket options which are considered herein. This work extends previous results of basket options to dimensions $n \geq 3$ and more generally, to payoff functions that satisfy Lipschitz continuity.}

\abstract{Exponential L\'evy processes have been used for modelling financial derivatives because of their ability to exhibit many empirical features of markets. Using their multidimensional analog, a general analytic pricing formula is obtained, allowing for the direct valuation of multi-asset options on $n \in \z^+$ risky assets. By providing alternate expressions for multi-asset option payoffs, the general pricing formula can reduce to many popular cases, including American basket options which are considered herein. This work extends previous results of basket options to dimensions $n \geq 3$ and more generally, to payoff functions that satisfy Lipschitz continuity.}

\section{Introduction}
\label{intro}
Thorough empirical studies have shown that financial markets exhibit skewness, kurtosis, an absence of autocorrelation in price increments, finite variance, aggregational normality, and have an ability to change discontinuously \cite{tankov2003financial}. To capture these features, exponential L\'evy processes are chosen to represent the asset model driving option prices. In much of the related literature, the focus is on solving option prices in terms of asset log-prices. For example, it is well known that by using variable transformations in the Black-Scholes equation, the problem reduces to solving a diffusion equation. However, with the Mellin transform one can circumvent this and solve the partial differential equation (PDE) directly. Despite this, the Mellin transform has only recently been considered in a financial context. 

In 2002, Cruz-B\'aez and Gonz\'alez-Rodr\'iguez pioneered the method of using Mellin transforms to solve the associated PDE for a European call option \cite{zbMATH01780547}. In 2004, Panini and Srivastav use the Mellin transform method to solve for the European put, American put, European basket put ($n=2$ underlying assets), and American basket put ($n=2$) \cite{panini2004option, PaniniS04}. In all of these cases, volatility is assumed constant, dividends are omitted, and the underlying asset model is assumed to be geometric Brownian motion (GBM). By incorporating a general L\'evy process, the associated problem becomes a partial integro-differential equation (PIDE), and the solutions to the aforementioned models can be obtained as special cases.

Thus, the objective of this manuscript is to provide a general option pricing model in the context of Mellin transforms. We derive an analytic formula for multi-asset European and American options with $n \in \z^+$ underlying risky assets represented by multidimensional exponential L\'evy processes. Since the only restriction placed on the payoff function is Lipschitz continuity, many other well-known payoffs can be applied to the formula. As an application, an alternative expression for basket payoffs is given, enabling the reduction to popular option formulas such as the generalized Black-Scholes-Merton. However, numerous other special cases are attainable.

This manuscript is organized as follows. Section 2 introduces the neccessary mathematical terminology used to construct and solve the pricing formula. In section 3, the pricing PIDE is solved in the general case, while section 4 provides an application to American basket options.

\section{Preliminaries}
\label{prelim}
To state the asset pricing model and derive the associated option formula, we must introduce L\'evy processes and the Mellin transform. The following overview considers their multidimensional cases, which are necessary to simulate multiple underlying assets.
\subsection{L\'evy processes}
\label{levy}
A L\'evy process $L$ is a stochastic process with independent and stationary increments \cite{kyprianou2006introductory}. In the multidimensional case, it has a representation given by $L_t=(L_{t1},...,L_{tn})'$. Let $\bar{u} \in \R^n$, excess return $\bar{\mu} \in \R^n$, covariance matrix $\Sigma \in \R^{n \times n}$ be symmetric positive definite, and $\nu$ be a measure concentrated on $\R^n / \{0\}$.  A probability law $\eta$ of a real-valued random variable $L$ has characteristic exponent $\Psi(\bar{u}):= - \frac{1}{t} \log ( \Ex [e^{i\bar{u}L_t} ])$ such that,
\begin{align} \label{cf}
\Phi(\bar{u};t)= \displaystyle\int_{\R^n} e^{i \bar{u}' y} \eta(dy) = e^{- t\Psi(\bar{u})} 
\end{align}
iff there exists a triplet $(\bar{\mu},\Sigma,\nu)$ such that,
\begin{align} \label{lkn}
\Psi (\bar{u})=-i   \bar{u}' \bar{\mu} + \frac{1}{2}  \bar{u}' \Sigma \bar{u} + \int\limits_{\R^n } (1-e^{i \bar{u}' y}+ i \bar{u}' y \mathds{1}_{(|y|<1)}) \nu(dy).
\end{align}
Equation \eqref{lkn} is known as the L\'evy-Khintchine formula \cite{kyprianou2006introductory}. Alternatively, L\'evy processes can be expressed by their L\'evy-It\^o decomposition, which extends naturally to higher dimensions as,
\begin{align} \label{itodec}
L_{ti} = \mu_i t + \sigma_i W_{ti} + \int\limits_0^t \int\limits_{|y| \geq 1} y \eta_{L_i}(ds,dy) + \int\limits_0^t \int\limits_{|y|<1} y(\eta_{L_i} - \nu_{L_i})(ds,dy)
\end{align}
for $1 \leq i \leq n$. The L\'evy measure $\nu$ of ${L}_t$ satisfies,
\begin{align}
\int_{\R^n} \min(1,y^2) \nu(dy) < \infty
\end{align}
with $\nu(\{0\})=0$. To incorporate correlation between the Brownian motions of the process, let ${\sigma} \in \R^{n \times n}$ be a diagonal matrix of volatilities $\sigma_i$ for $1 \leq i \leq n$ and ${\rho} \in \R^{n \times n}$ be a correlation matrix with $\rho_{ij}= \corr(dW_i,dW_j) \in [-1,1]$ such that $\Sigma=  {\sigma} {\rho} {\sigma}$.
\subsection{The Mellin transform}
In order to solve the associated PIDE for the option price, the multidimensional Mellin transform is required. For a thorough treatment of the transform and its properties, see \cite{brychkov,sneddon}.
\begin{definition} [Mellin Transform] \label{mel} Let $\bar{x}=(x_1,..,x_n)'$ and $\bar{w}=(w_1,..,w_n)'$. For a function $f(\bar{x}) \in \R^{n+}$ the multidimensional Mellin transform is the complex function, \footnote{$\bar{x}^{\bar{w}-1} d\bar{x}$ is treated as $ \prod_{j=1}^n x_i^{w_i-1} dx_i$.}
\begin{eqnarray}
 \hat{f}( \bar{w}) := \ML \lbrace f(\bar{x}) ; \bar{w} \rbrace= \int_{\R^{n+}} f( \bar{x}) \bar{x}^{\bar{w}-1} d\bar{x}.
\end{eqnarray}
\end{definition}
The largest domain in which $\hat{f}$ is analytic is known as the {\sl fundamental strip} and is often denoted by $ \langle \bar{a}^{0},\bar{a}^{\infty} \rangle$. Consider the inverse scenario; the Mellin transform of a function is known, and one wishes to recover the original function. For a function $\hat{f}(\bar{w}) \in \C^n$ it can be shown under general conditions that an inverse $f(\bar{x}) \in \R^{n+}$ not only exists, but is also unique (for a given fundamental strip).
\begin{theorem} [Mellin Inversion Theorem] \label{multimelinv} Let $\bar{w}=(w_1,..,w_n)'$, $\bar{x}=(x_1,...,x_n)'$, and $\hat{f}(\bar{w}) \in \C^n$ be analytic on $\gamma = \overset{n}{ \underset{j=1}{ \times}} \gamma_j$ defined by $\gamma_j = \{ a_j+ib_j : a_j \in \R,  b_j = \pm \infty \}$ with $a_j \in \Re(w_j)= \langle a_j^{0},a_j^{\infty} \rangle$. Suppose $\int_\R \hat{f}(\bar{a}+i\bar{b}) d\bar{b}$ is absolutely convergent. Then,
\begin{align}
f(\bar{x}) = \ML^{-1} \lbrace  \hat{f}(\bar{w}) ; \bar{x} \rbrace = (2 \pi i)^{-n} \displaystyle\int_\gamma \hat{f}(\bar{w}) \bar{x}^{-\bar{w}} d\bar{w}.
\end{align}
\end{theorem} 
If $f(\bar{x})= \prod_{j=1}^n f_j(x_j)$, then $\hat{f}(\bar{w})= \prod_{j=1}^n \ML \lbrace f_j(x_j) ; w_j \rbrace$. This allows the well developed properties of the univariate transform to be used to solve the multidimensional case.

\section{Partial Integro-Differential Equation for Option Pricing}
\label{PIDE}
For $1 \leq i \leq n$ and risk-free rate $r>0$, consider the exponential L\'evy model $S_{ti}=S_{0i}e^{rt+L_{ti}}$ as the asset price process for a multi-asset option on a filtered probability space $(\Omega, \mathcal{F}, \mathbb{P})$. By the hypothesis of no-arbitrage, there exists a martingale measure $\mathbb{Q}$ equivalent to $\mathbb{P}$ \cite{tankov2003financial}. Since the discounted asset $(e^{-rt} \bar{S}_t)_{t \in [0,T]}$ is a martingale under $\mathbb{Q}$, it follows that $e^{{L}_{t}}$ is also a martingale. This leads to the following condition on the L\'evy process of $\bar{S}_t$:
\begin{eqnarray} \label{drift}
\mu_i (\sigma_i, \nu) = -r- \frac{\sigma_i^2}{2} - \int\limits_{\R^n } (e^y-1-y \mathds{1}_{|y| <1} ) \nu_{L_i}(dy).
\end{eqnarray}
 European option prices with fixed maturity $T < \infty$ can be calculated from a discounted expectation of their payoff function $\theta(\bar{S})$ under $\mathbb{Q}$. i.e.
\begin{eqnarray} \label{europay}
V(\bar{S}, t)=\Ex_\mathbb{Q} \big( e^{-r(T-t)} \theta(\bar{S}(T)) | \bar{S}_t=\bar{S} \big).
\end{eqnarray}
\begin{theorem} [\protect{\cite[Theorem 4.2.]{15242739}}]
Let $L$ be a L\'evy process with state space $\R^n$ and characteristic triplet $(\bar{\mu}, \Sigma, \nu)$. Assume that the function $V(\bar{S}, t)$ in \eqref{europay} satisfies $V(\bar{S},t) \in C^{2,1} \big( \R^{n+} \times (0,T) \big)  \cap C^0 ( \R^{n+} \cup \{0\} \times [0,T] \big)$. Then, $V(\bar{S}, t)$ is a classical solution of the backward Kolmogorov equation:
\begin{eqnarray} \label{kolx} \nonumber
\frac{\partial V}{\partial t} &+ \frac{1}{2} \displaystyle\sum_{\substack{ i,j=1 \\ i \neq j}}^n \rho_{ij} \sigma_i \sigma_j S_i S_j  V_{S_i S_j} + \frac{1}{2} \displaystyle\sum_{i=1}^n \sigma_i^2 S_i^2 V_{S_i S_i}  + r \displaystyle\sum_{i=1}^n S_i  V_{S_i} - r V \\ &+ \displaystyle\int_{\R^{n} }  \bigg[V(\bar{S}e^{ y}) - V - \displaystyle\sum_{i=1}^n  (e^{ y_i}-1)S_i V_{S_i} \bigg] \nu(dy) =0
\end{eqnarray}
on $(0,T) \times \R^{n+}$ where $V(\bar{S}e^y) :=V(S_1e^{y_1}, ... , S_2 e^{y_2}, t)$ and the terminal condition is given by $V(\bar{S}, T)= \theta(\bar{S})$. \footnote{In the proof, $\theta$ is required to be Lipschitz continuous.}
\end{theorem}
Note that the above PIDE formulation only considers the European case. American options differ in that they can be exercised at any time $t<T<\infty$. It is known in the American case that a decomposition exists where the option value can be represented as a sum of a European option and an early exercise premium i.e. $V_A(\bar{S}, t)  = V_E (\bar{S}, t)+V_{EEP} (\bar{S}, t)$ \cite{jarrow,kimam}. Consider the ansatz that the solution to an inhomogenous PIDE \eqref{kolx} solves the American case. For notational simplicity, define $\mathcal{L}[V(\bar{S})]$ and let $f=f(\bar{S},t)$ so that the PIDE becomes,
\begin{eqnarray} \label{kols2}
\frac{\partial V}{\partial t}+\mathcal{L}[V(\bar{S})] =f ; \hspace{5mm} V(\bar{S}, T)=\theta(\bar{S}).
\end{eqnarray}
Let $V(\bar{S})$ be bounded when $\bar{S} \to \infty$ and $V(0, t) = Ke^{-r(T-t)}$. The solution to \eqref{kols2} is given by, 
\begin{theorem} \label{main} Let $\bar{S}=(S_1,...,S_n)'$, $\bar{w}=(w_1,...,w_n)'$, $0 \leq \tau \leq T$, and $0 < K, T, S_j < \infty$ for all $1 \leq j \leq n$. For Lipschitz payoff ${\theta}(\bar{S})$, the value of a multi-asset option $V(\bar{S},\tau)$ on $n$ exponential L\'evy-driven assets is given by,
\begin{eqnarray} \label{value}
V(\bar{S}, \tau) &=  e^{-r \tau} \ML^{-1} \Big\lbrace  \hat{\theta} \Phi(\bar{w}i, \tau) \Big\rbrace +  \ML^{-1} \Big\lbrace  \displaystyle\int_0^\tau \hat{f} \Phi(\bar{w}i, \tau-s) e^{-r(\tau-s)}  ds  \Big\rbrace 
\end{eqnarray}
where $\Phi(\cdot)$ is the characteristic function of the L\'evy process in \eqref{cf}.
\end{theorem}
\begin{proof}
Consider the homogenous case. Apply Definition \ref{mel} to $\mathcal{L}[V(\bar{S})]$ to obtain an expression for $\mathcal{L}[V(\bar{w})]$ where $\bar{w} = (w_1,...,w_n)'$ are complex Mellin variables. Since $\ML \lbrace V_t  ; \bar{w}  \rbrace= \hat{V}_t$ by independence, an expression for the multidimensional Mellin transform of the homogenous PIDE can be formed:
\begin{eqnarray} \label{sov}
\hat{V}_t(\bar{w},t)= -Q(\bar{w}) \hat{V}(\bar{w},t) \hspace{5mm} {\text{ where }} \hspace{5mm} Q(\bar{w})&= \frac{ \mathcal{L}[\hat{V}(\bar{w})] }{\hat{V}} = -\Psi(\bar{w}i)  - r
\end{eqnarray}
via \eqref{drift} and \eqref{lkn}. Using the initial condition and solving for \eqref{sov} yields $\hat{V}(\bar{w},t)=\hat{\theta}(\bar{w})e^{Q(\bar{w})(T-t)}$. Let $\tau=T-t$. \footnote{The problem must be recast as an initial value problem to use Duhamel's principle.}  Duhamel's principle solves the inhomogenous problem by considering the contribution of $f(\bar{S},s)$ when $s < t$. The result follows from \eqref{cf} and Theorem \ref{multimelinv}. \qed
\end{proof}

When $\hat{f}=0$, \eqref{value} is the value of a multi-asset European option. When $\hat{f} \neq 0$, \eqref{value} is the value of a multi-asset American option. Hence, the second term of $V(\bar{S},\tau)$ represents the early exercise premium. To implement this formula, one must know $(i)$ the Mellin payoff $\hat{\theta}(\bar{w})$, $(ii)$ the characteristic function (or exponent) of the L\'evy process, and $(iii)$ the Mellin transform of the exercise boundary $\hat{f}(\bar{w})$. The subsequent section considers these three components for American basket put options with an asset driven by GBM.

\section{Application to American Basket Options}
Consider the case where a put option on $n$ Brownian motion assets has an arithmetic basket payoff function given by $\theta_{P}(\bar{S})=\max (K-\sum_{j=1}^n S_j )=\big(K-\sum_{j=1}^n S_j )^+$. \footnote{An explicit expression for the put of a basket payoff in terms of Fourier transforms is obtained in \cite{HurdZ10} when $K=1$. With an appropriate change of variables ($\bar{u}=\bar{w}i$) their formula can be verified using Proposition \ref{basketput}. } One can prove the following alternate expression by induction and deduce the fundamental strip by the argument of its gamma functions.
\begin{proposition} \label{basketput}
Let $\bar{S}=(S_1,...,S_n)'$, $\bar{w}=(w_1,...,w_n)'$, $0 \leq \tau \leq T$, and $0 < K,T, S_j < \infty$ for all $1 \leq j \leq n$. For $\Re(w_j) > 0$, the put payoff function for a basket option is,
\begin{eqnarray} \label{bp}
\theta_{P}(\bar{S})=  \ML^{-1} \lbrace \hat{\theta}(\bar{w}) \rbrace = \ML^{-1} \bigg\lbrace  \frac{ \beta_n(\bar{w}) K^{1+ \sum \bar{w}}} {(\sum \bar{w}) (1 + \sum \bar{w})} \bigg\rbrace
\end{eqnarray}
where $\beta_n(\bar{w})$ denotes the multinomial beta function. \footnote{ $\beta_n(\bar{w}) = \frac{ \prod \Gamma(\bar{w})}{\Gamma (\sum \bar{w})}$ where $\Gamma(\cdot)$ denotes the complex gamma function.}
\end{proposition}
Knowing the exercise boundary allows us to solve for the Mellin transform of $f$:
\begin{eqnarray} \label{eb}
\hat{f}(\bar{w})=\frac{- rK \beta_n(\bar{w})}{\sum \bar{w}}(S^*(s))^{\sum \bar{w}}
\end{eqnarray}
where $S^*$ denotes the {\em critical asset price}. Furthermore, since the characteristic function of a L\'evy process exists by \eqref{cf}, the characteristic function of returns $\Phi(\cdot)$ driven by an $n$-dimensional Brownian motion is known. Hence, an expression for the American basket put option on $n$ risky assets driven by GBM is obtained using $ \hat{\theta}(\bar{w})$ in \eqref{bp}, \eqref{eb}, and $\Phi(\cdot)$ in accordance with Theorem \ref{main}. By setting $n=2$ we obtain the American basket option formula derived in \cite{panini2004option, PaniniS04}. If $n=1$ we obtain the Black-Scholes-Merton American option \cite{kimam, kwok2008}. Similarly, an expression for the call option can be obtained by use of the generalized put-call parity relationship \cite{zbMATH01780547,tankov2003financial}.

\bibliographystyle{plain}
\bibliography{Refs}

\end{document}